\documentclass[aip,jmp,reprint,onecolumn,draft]{revtex4-1}
\usepackage{amsmath}
\usepackage{amsfonts}
\usepackage{amssymb}
\usepackage{amsthm}
\newcommand{\ignore}[1]{}

\newcommand{\R}{{\mathbb{R}}}

\newtheorem{thm}{Theorem}[section]

\newtheorem{cor}[thm]{Corollary}

\newtheorem{lemma}[thm]{Lemma}

\theoremstyle{definition}
\newtheorem{defn}[thm]{Definition}

\theoremstyle{remark}
\newtheorem{remark}[thm]{Remark}

\begin{document}

\title{Unentangled Measurements and Frame Functions}

\author{Ji\v{r}\'{\i} Lebl}
\email{lebl@math.okstate.edu}
\thanks{The first author was in part supported by NSF grant DMS-1362337 and
Oklahoma State University's DIG and ASR grants.}
\affiliation{Department of Mathematics, Oklahoma State University,
Stillwater, OK 74078, USA}

\author{Asif Shakeel}
%\address{Department of Coolness, University of Wherever, Somewhere Warm,
%99999}
\email{asif.shakeel@gmail.com}
\noaffiliation

\author{Nolan Wallach}
\email{nwallach@ucsd.edu}
\affiliation{Department of Mathematics,
University of California, San Diego,
9500 Gilman Drive, La Jolla, CA 92093-0112, USA}

\date{\today}

\begin{abstract}
Gleason's theorem asserts the equivalence of von Neumann's density operator
formalism of quantum mechanics and frame functions, which are functions on
the pure states that sum to $1$ on any orthonormal basis of Hilbert space of
dimension at least $3$. The {\it unentangled frame functions} are initially
only defined on unentangled (that is, product) states in a multi-partite
system. The third author's {\it Unentangled Gleason's Theorem} shows that
{\it unentangled frame functions} determine unique density operators if and
only if each subsystem is at least $3$-dimensional. In this paper, we
determine the structure of unentangled frame functions in general. We first
classify them for multi-qubit systems, and then extend the results to
factors of varying dimensions including countably infinite dimensions
(separable Hilbert spaces). A remarkable combinatorial structure emerges,
suggesting possible fundamental interpretations.
\end{abstract}


\maketitle

%\enlargethispage{\baselineskip}

%%%%%%%%%%%%%%%%%%%%%%%%%%%%%%%%%%%%%%%%%%%%%%%%%%%%%%%%%%%%%%%%%%%%%%%%%%%%

\section{Introduction} \label{section:intro}

In von Neumann's~\cite{vn:mgdq} approach to quantum mechanics, the formalism
assumes a factor algebra with an endowed trace function. The mixed states in
this set up are the self-adjoint, positive elements of trace $1$ looked upon
as defining a measure on the set of self-adjoint, idempotent elements. This
led Mackey to ask the natural question: Is every such measure given by a
mixed state? Gleason gave an affirmative answer to this question in the case
of factors of type $\text{I}_n$ with $n > 2$ (i.e. the bounded operators on
a separable Hilbert space). In his approach to his proof, he introduced the
notion of frame function (a non-negative function on the pure states that
sums to $1$ on an orthonormal basis of the Hilbert space) which is easily
seen to be equivalent to that of measure on the set of self-adjoint
idempotents.  Gleason's theorem shows that such a function, $f$, must be of
the form $f(v) = \left\langle v|A|v\right\rangle$ for $v$ any unit norm
vector in  the Hilbert space and $A$ a non-negative, self-adjoint operator
of trace $1$ if the dimension of the Hilbert space is at least $3$. In many
contexts of quantum information theory that deal with state spaces of many
independent particles, the only pure states that are sampled are product (or
unentangled) states. This led the last named author~\cite{W:Gleason} to ask
if such a sampling of only unentangled states would allow for more general
frame functions. In this paper we classify all the unentangled frame
functions, for an arbitrary but finite number of tensor factors, including
those of countably infinite dimension (separable Hilbert spaces), completing
the treatment of the unentangled frame functions.  A thorough account  of  quantum measurement  theory   in the setting of operator algebras, including   Gealson's theorem and its variants, is in~\cite{h:qmt} (J. Hamhalter).

 The organization  of the rest of the paper is  as follows. We  first introduce the idea of an {\it unentangled frame function} in Section~\ref{section:framefunc}. In Section~\ref{section:fundom}, we choose the fundamental domain in $\mathbb{P}^{1}(\mathbb{C})$ that we use  to identify a vector with its orthogonal.
In Section~\ref{section:framefuncsn} we  classify  all multi-qubit frame functions. Section~\ref{section:genframefun} generalizes this classification further to include the case when some of the tensor factors are of dimensions at least $3$ and at most countable (separable Hilbert spaces).
Section~\ref{section:conc} concludes our discussion with some remarks.

%%%%%%%%%%%%%%%%%%%%%%%%%%%%%%%%%%%%%%%%%%%%%%%%%%%%%%%%%%%%%%%%%%%%%%%%%%%%

\section{Unentangled frame functions}
\label{section:framefunc}

Let us take a brief look at unentangled Gleason setup as in~\cite{W:Gleason}, and make more precise the ideas from the Introduction. Let

\[
\mathcal{H}=H_{1}\otimes H_{2}\otimes\cdots\otimes H_{n}
\]
with $\dim H_{i}\geq2$. Technically we should be looking at the completed
tensor product if two or more of the factors are infinite dimensional.
However, since we will only be looking at product vectors this will not be
necessary. Applying a permutation of the factors we may assume that the first
$k$ of the $H_{i}$ are of dimension $2$ and all of the rest have dimension
$>2$. Thus
\[
\mathcal{H}=\otimes^{k}\mathbb{C}^{2}\otimes H_{k+1}\otimes\cdots\otimes
H_{n}
\]
with $\dim H_{i}>2$ and if $n=k$ then by convention the last factor is
$\mathbb{C}$.
$\mathcal{H}$ is given  the tensor product Hilbert
structure,  $\left\langle \ldots |\ldots \right\rangle $. A vector in $\mathcal{H}$
is called   unentangled  (a product vector) if it is a tensor product
of unit vectors, one from each $H_{i}$ factor. Two such vectors $v_{1}\otimes
v_{2}\otimes\cdots\otimes v_{n}$ and $w_{1}\otimes w_{n}\otimes\cdots\otimes
w_{n}$ are orthogonal 
\[
\left\langle v_{1}\otimes v_{2}\otimes\cdots\otimes v_{n}|w_{1}\otimes
w_{2}\otimes\cdots\otimes w_{n}\right\rangle =0
\]
if and only if there is at least one $i$ with $\left\langle v_{i}
|w_{i}\right\rangle =0$.  An Unentangled Othonormal Basis (UOB) $\{u_i\}$ is a basis  of $\mathcal{H}$ consisting of orthogonal
(unit norm)
unentangled vectors. Let $\Sigma$ be the set of all unentangled vectors in $\mathcal{H}$. 
\begin{defn} \label{ffdef}
An unentangled frame function is a map
\begin{equation*}
f \colon \Sigma \to \R_{\geq 0},
\end{equation*}
  such that for every UOB $\{u_i\}$, 
\begin{equation} \label{ffsum}
\sum_i f(u_i) = 1.
\end{equation}
\end{defn}

In the next two sections we will be dealing with the multi-qubit case, so we specialize the above to prepare for it. Let  $\mathcal{H}_{n} = {\otimes}^{n}\mathbb{C}^{2}$ be the space of $n$ qubits.
 A UOB is then a basis $\{u_{0}, u_{1}, \ldots,
u_{2^{n}-1}\}$ of $\mathcal{H}_{n}$ consisting of orthogonal
(unit norm)
unentangled vectors.
 
\section{A fundamental domain} \label{section:fundom}

Let $\sigma:\mathbb{C}^{2}\rightarrow\mathbb{C}^{2}$ be defined by
$v=(x,y)\longmapsto \hat v = (-\bar{y},\bar{x})$.
We note that $\left\langle v|\sigma
v\right\rangle =0$ and if $v$ is a state then up to phase $\sigma v$ is the
unique state perpendicular to $v$. The $\sigma$ induces a map of
$\mathbb{P}^{1}(\mathbb{C})$ to itself which we also denote by $\sigma$. We
have the standard map
\[
\mathbb{P}^{1}(\mathbb{C})\rightarrow\mathbb{C\cup\infty}
\]
given by
\[
(x,y)\longmapsto\frac{x}{y}.
\]
We note that under this identification as a map from $\mathbb{C\cup\infty}$ to
itself
\[
\sigma z=-\frac{1}{\bar{z}}.
\]
In particular, on $S^{1}$ it is given by $z\longmapsto-z$.

Here is a simple fundamental domain for $\sigma$:%
\[
F=\{z\in\mathbb{C}\mid |z|<1\}\cup\{z\in\mathbb{C}\mid |z|=1,\operatorname{Im}%
z>0\}\cup\{1\}.
\]
A frame function $f$
in a single qubit space  corresponds to an arbitrary function
$\phi \colon F \to [0,1]$ by letting $f(a) = \phi(a)$ if $a \in F$ and since the sum of the values of $f$ over orthogonal vectors is $1$, 
$f(a) = 1-\phi(\sigma a)$ if $a \notin F$.

%%%%%%%%%%%%%%%%%%%%%%%%%%%%%%%%%%%%%%%%%%%%%%%%%%%%%%%%%%%%%%%%%%%%%%%%%%%%

\section{Unentangled frame functions in $n$ qubits}
\label{section:framefuncsn}

Let
$\mathcal{H}_{n}$ denote $n$--qubit space.
We choose a fundamental domain $F$
for the map $[z]\rightarrow\sigma([z])=[\hat{z}]$ in one qubit. We set
$F_{n}=F\otimes F\otimes\cdots\otimes F$ ($n$ copies). Let $\Omega
_{n}=\{1,\ldots,n\}$. If $J\subset\Omega_{n}$ is a subset we define $\sigma
_{J}=T_{1}\otimes T_{2}\otimes\cdots\otimes T_{n}$ with $T_{j}=\sigma$ if
$j\in J$ and $T_{j}=\mathbb{I}$ (identity operator),
if $j\notin J$. We set $J^{c}=\Omega_{n}-J$. We note
\[
\Sigma=\bigcup_{J\subset\Omega_{n}}\sigma_{J}(F_{n})
\]
is a disjoint union. If $z=z_{1}\otimes\cdots\otimes z_{n}\in F_{n}$, if
$J\subset\Omega_{n}$ and if $J^{c}=\{i_{1},\ldots,i_{k}\}$ with $i_{1}%
<i_{2}<\cdots<i_{k}$ then we set $\tau_{J}(z)=(z_{i_{1}},\ldots,z_{i_{k}})$ if
$J\neq\Omega_{n}$, if $J=\Omega_{n}$ then use the symbol $\omega$ for the value.

\begin{lemma} \label{lemma1}
If $f$ is a function on the product states of $\mathcal{H}_{n}$ such that $f$
sums to a fixed constant $c$ on all UOB's then for each $J\subset
\Omega_{n}\ $there exists a function $\phi_{J}$ on $F^{n-|J|}$ (direct product
of $n-|J|$ copies of $F$, $F^{0}=\{\omega\},\phi_{\Omega_{n}}(\omega)=c$) such
that if $z\in F_{n}$ then%
\[
\sum_{L\subset J}f(\sigma_{L}(z))=\phi_{J}(\tau_{J}(z)).
\]
\end{lemma}

\begin{proof}
After permuting the factors we may assume that $J=\{1,\ldots,j\}$, thus the sum on
the left hand side is
\[
\sum_{I\subset\{1,\ldots,j\}}f(\sigma_{I}(z_{1}\otimes\cdots\otimes z_{j}\otimes
z_{j+1}\otimes\cdots\otimes z_{n}))
\]
call it $b$, and note that $\sigma_{I}$ acts only  on $z_{1}\otimes\cdots\otimes z_{j}$ as $\sigma$ and as $\mathbb{I}$ everywhere else. Observing that if the set
\[
Z=\{\sigma_{I}(z_{1}\otimes\cdots\otimes z_{j}\otimes z_{j+1}\otimes
\cdots\otimes z_{n}) \mid I\subset J\}
\]
is extended to a $UOB$ by adjoining elements $u_{1},\ldots,u_{r}$ with
$r=2^{j}(2^{n-j}-1)$ then $\sum f(u_{i})=1-b$ no matter how we found the
extension. Also $Z$ is an orthonormal basis of $\otimes^{j}\mathbb{C}%
^{2}\otimes z_{j+1}\otimes\cdots\otimes z_{n}$ which only depends on
$z_{j+1}\otimes\cdots\otimes z_{n}$. This proves the result.
\end{proof}

Now applying inclusion exclusion we have in the notation of the previous lemma

\begin{lemma} \label{lemma2}
If $z\in F_{n}$ and if $f$ is a function on the product states satisfying the
hypotheses of Lemma~\ref{lemma1} then%
\[
f(\sigma_{J}(z))=\sum_{L\subset J}(-1)^{|J-L|}\phi_{L}(\tau_{L}(z)).
\]
\end{lemma}

\begin{proof}
Inclusion exclusion says: Let $\alpha$ and $\beta$ be functions from the set
of all subsets of $\Omega_{n}$ to $\mathbb{C}$ and such that if $J\subset
\Omega_{n}$ then
\[
\sum_{L\subset J}\alpha(L)=\beta(J) ,
\]
then%
\[
\alpha(J)=\sum_{L\subset J}(-1)^{|J-L|}\beta(L).
\]
(c.f. Rota~\cite{r:mobius}).
The lemma follows from this assertion by taking $\alpha(J)=f(\sigma_{J}(z))$ and $\beta(J)=\phi_{J}(\tau_{J}(z))$.
\end{proof}

\begin{thm} \label{theorem3}
If $L\subsetneqq\Omega_{n}$ let $\phi_{L}$ be a real valued function on
$F_{n-\left\vert L\right\vert }$. Assume that $\phi_{\Omega_{n}}(\omega)=c$
then the function $f:\mathcal{H}_{n}\rightarrow\mathbb{R}$ defined by%
\[
f(\sigma_{J}(z))=\sum_{L\subset J}(-1)^{\left\vert J-L\right\vert }\phi
_{L}(\tau_{L}(z))
\]
for $J\subset\Omega_{n}$ and $z\in F_{n}$ satisfies
\[
\sum_{i=1}^{2^{n}}f(z_{i})=c\text{.}
\]
if $z_{1},\ldots,z_{2^{n}}$ is a UOB.
\end{thm}

\begin{proof}
We first note that if $z\in F_{n-1}$ and $J\subset\{2,\ldots,n\}$ then

1. $f(a\otimes\sigma_{J}(z))+f(\hat{a}\otimes\sigma_{J}(z))$ is independent of
$a$.

Indeed,
\[
f(a\otimes\sigma_{J}(z))=\sum_{L\subset J}(-1)^{\left\vert J|-|L\right\vert
}\phi_{L}(\tau_{L}(a\otimes z))
\]
and (if $1\in L$ write $L=\{1\}\cup L^{\prime}$ with $L^{\prime}\subset
\{2,\ldots,n\}$)
\[
f(\hat{a}\otimes\sigma_{J}(z))=f(\sigma_{J\cup\{1\}}(a\otimes z))=\sum
_{L\subset J\cup\{1\}}(-1)^{|J|+1-|L|}\phi_{L}(\tau_{L}(a\otimes z))=
\]%
\[
-\sum_{L\subset J}(-1)^{\left\vert J|-|L\right\vert }\phi_{L}(\tau
_{L}(a\otimes z))+\sum_{L^{\prime}\subset J}(-1)^{|J|-|L^{\prime}|}%
\phi_{L^{\prime}\cup\{1\}}(\tau_{L^{\prime}\cup\{1\}}(a\otimes z)).
\]
We therefore have
\[
f(a\otimes\sigma_{J}(z))+f(\hat{a}\otimes\sigma_{J}(z))=\sum_{L^{\prime
%JL: removed :
}\subset J}(-1)^{|J|-|L^{\prime}|}\phi_{L^{\prime}\cup\{1\}}(\tau_{L^{\prime
}\cup\{1\}}(a\otimes z)).
\]
This proves 1.\ since $\tau_{L^{\prime}\cup\{1\}}(a\otimes z)$ is independent
of $a$.

We will now prove the theorem by induction on $n$. If $n=1$ the assertion is
that if $a\in F$ then $f(a)+c-f(a)=c$. So the result is true for $n=1$. Now
assume the result for $n-1\geq1$. We now prove it for $n.$ If $a\in F$ then
define $f_{a}(z)=f(a\otimes z)$. Then if $J\subsetneqq\Omega=\{2,\ldots,n\}$ and
$z\in F_{n-1}$ we have
\[
f_{a}(\sigma_{J}(z))=\sum_{L\subset J}(-1)^{\left\vert J-L\right\vert }%
\phi_{L}(\tau_{L}(a\otimes z))
\]
and%
\[
f_{a}(\sigma_{\Omega}(z))=\phi_{\Omega}(a).
\]
Thus the inductive hypothesis implies that $f_{a}$ satisfies%
\[
\sum_{i=1}^{2^{n}-1}f_{a}(z_{i})=\phi_{\Omega}(a)\text{.}
\]
for any UOB $\mathcal{\{}z_{1},\ldots,z_{2^{n-1}}\}$ of $\mathcal{H}_{n-1}$.

We are now ready to prove the theorem. Let $\mathcal{B}=\{z_{1},\ldots,z_{2^{n}%
}\}$ be a UOB. Then Theorem $6$ in~\cite{W:Gleason} implies that there
exist $a_{1},\ldots,a_{r}\in F$, $V_{1},\ldots,V_{r}$ orthogonal subspaces of
$\mathcal{H}_{n-1}$ such that
\[
\mathcal{H}_{n-1}=V_{1}\oplus \cdots \oplus V_{r}
\]
and $u_{ij}$ and $v_{ij},j=1,\ldots,d_{i}$ orthonormal basis of
$V_{i}$ consisting of product vectors such that
\[
\mathcal{B}=\{a_{i}\otimes u_{ij}\mid i=1,\ldots,r,j=1,\ldots,d_{i}\}\cup\{\hat{a}%
_{i}\otimes v_{ij}\mid i=1,\ldots,r,j=1,\ldots,d_{i}\}.
\]
For each $i$ we apply the inductive hypothesis to $f_{a_{i}}$ and find that%
\[
\sum_{j=1}^{d_{i}}f_{a_{i}}(u_{i,j})
\]
depends only on $a_{i}$ and $V_{i}$ and not on the particular orthonormal
basis of $V_{i}$. Thus we can replace $u_{i,j}$ with $v_{i,j}$ without changing
the sum. Now%
\[
f(a_{i}\otimes v_{i,j})+f(\hat{a}_{i}\otimes v_{ij})
\]
is independent of $a_{i}$ $by$ $1.$ Thus we can replace all of the $a_{i}$
with a fixed element $a\in F$ without changing the sum. Thus the sum is given
by%
\[
\sum_{ij}f(a\otimes v_{ij})+\sum_{ij}f(\hat{a}\otimes v_{ij}).
\]
We now observe that if we define%
\[
g(z)=f(a\otimes z)+f(\hat{a}\otimes z)
\]
then (see the proof of 1.)%
\[
g(\sigma_{J}(z))=\sum_{L^{\prime}\subset J}(-1)^{|J|-|L^{\prime}|}\phi_{L^{\prime}\cup\{1\}}(\tau_{L^{\prime}\cup\{1\}}(a\otimes z)).
\]
and
\[
\sum_{J\subset\Omega}g(\sigma_{J}(z))=c
\]
for all $z\in F_{n-1}$. Finally we can apply the inductive hypothesis to
replace the basis $\{v_{ij}\}$ by $\{\sigma_{J}(z) \mid J\subset\Omega\}$ with
$z\in F_{n-1}$ and the theorem is proved.
\end{proof}
A consequence of the proof of this theorem is that it suffices to specify a frame function on the highest dimensional component of the space of UOBs. This is described in\cite{lsw:lduobag}.  A generic UOB in this component is recursively defined as follows. 
\begin{defn}
\begin{equation*}
\mathcal{B} = 
\{
a \otimes \mathcal{B}_1 ,  \hat{a} \otimes \mathcal{B}_2
\} , 
\end{equation*}
for an arbitrary $a \in F$, and where $\mathcal{B}_i, i=1, 2$ are again defined in the same manner as $\mathcal{B}$ for one less qubit.
\end{defn}

\begin{cor}
If $f$ is a function on the product state such that there exists a constant
$c$ and for every generic UOB, $\{z_{1},\ldots,z_{2^{n}}\}$we have%
\[
\sum f(z_{i})=c
\]
then the same is true for every UOB.
\end{cor}

\begin{proof}
The characterization of the generic UOB in~\cite{lsw:lduobag} makes it clear
that if $f$ sums to $c$ on all generic UOB then $f$ and the functions $\phi_{J}$
have the property in Lemma~\ref{lemma1}. Lemma~\ref{lemma2} and Theorem~\ref{theorem3} now complete the proof.
\end{proof}

Notice that the corollary is interesting only when $n \geq 3$, since when $n
< 3$, every UOB is part of a maximal dimensional family.

\section{General unentangled frame functions} \label{section:genframefun}

In this section we will give a complete description of unentangled frame
functions for separable Hilbert spaces of the form
\[
\mathcal{H}=H_{1}\otimes H_{2}\otimes\cdots\otimes H_{n}
\]
with $\dim H_{i}\geq2$. 

Let $f$ be an unentangled frame function on $\mathcal{H}$. Thus $f$ is a real
valued function product states such that there exists a scalar $c$ such that
if $\{u_{i}\}$ is a $UOB$ then $\sum f(u_{i})=c$.

If $z$ is a product state in $\mathcal{\otimes}^{k}\mathbb{C}^{2}$ then
setting $f_{z}(x)=f(z\otimes x)$ for $x$ a product state in $H_{k+1}%
\otimes\cdots\otimes H_{r}$, $f_{z}$ is an unentangled frame function on
$H_{k+1}\otimes\cdots\otimes H_{r}$ with%
\[
\sum f_{z}(u_{i})=c(z)
\]
for $\{u_{i}\}$ a $UOB$ of $H_{k+1}\otimes\cdots\otimes H_{r}$ (so
$c(z)=\sum_{i}f(z\otimes u_{i})$ depends only on $z$ and $f$). The unentangled
Gleason's theorem~\cite{W:Gleason}  implies that there exists $A(z)$ a trace class, self-adjoint
non-negative operator on $H_{k+1}\otimes\cdots\otimes H_{r}$ such that

1. $\mathrm{tr}A(z)=c(z)$.

2. $f_{z}(u)=\left\langle u|A(z)|u\right\rangle$.

Note that the proof of the unentangled
Gleason theorem  given in~\cite{W:Gleason}  does not use finite dimensionality and therefore applies to the context of this paper.

This proves the following reduction of the problem.

\begin{lemma}
Let $\mathcal{H}=\mathcal{\otimes}^{k}\mathbb{C}^{2}\otimes H_{k+1}%
\otimes\cdots\otimes H_{r}$ be as above. Then an unentangled frame function
for $\mathcal{H}$ is the restriction of one for $\mathcal{\otimes}%
^{k}\mathbb{C}^{2}\otimes H$ with $H$ the completion of $H_{k+1}\otimes
\cdots\otimes H_{r}$.
\end{lemma}

In light of this lemma we need only classify the unentangled frame functions on
Hilbert spaces of the form $\mathcal{\otimes}^{k}\mathbb{C}^{2}\otimes H$ with
$H$ a separable Hilbert space of whose dimension is not $2$. We note that if
$\dim H=1$ then a self-adjoint operator on $H$ is a real scalar. Let
$\mathcal{T}(H)$ be the space of trace class, self-adjoint operators on $H$.

\begin{thm}
Let for each $J\subset\{1,\ldots,k\}$, $\phi_{J}:F^{k-|J|}\rightarrow
\mathcal{T}(H)$ (here as usual, we set $F^{0}=\{\omega\}$). Let $z\in F^{k}$
and let $u$ be a state in $H$
\[
f(\sigma_{J}(z)\otimes u)=\sum_{L\subset J}(-1)^{|L-J|}\left\langle u|\phi
_{L}(\tau_{L}(z))|u\right\rangle .\qquad \overset{}{(\ast)}%
\]
Let
\[
c=\mathrm{tr}\phi_{\Omega_{k}}(\omega).
\]
If $\{w_{j}\}$ is a $UOB$ of $\mathcal{\otimes}^{k}\mathbb{C}^{2}\otimes H$
then
\[
\sum f(w_{j})=c.
\]
If $f$ is an unentangled frame function on $\mathcal{\otimes}^{k}%
\mathbb{C}^{2}\otimes H$ and if $J\subset\Omega_{k}$ we set for $z\in F^{k}$
and $u$ a state in $H$%
\[
\gamma_{J,z}(u)=\sum_{L\subset J}f(\sigma_{L}(z)\otimes u),
\]
then $\gamma_{J,z}(u)=\left\langle u|\phi_{J}(\tau_{J}(z))|u\right\rangle $ with
$\phi_{J}:F^{k-|J|}\rightarrow\mathcal{T}(H)$ and $f$ is given by $(\ast)$.
\end{thm}

\begin{proof}
If $f$ is an unentangled frame function on $\mathcal{\otimes}^{k}%
\mathbb{C}^{2}\otimes H$ and if $z$ is a state in $\mathcal{\otimes}%
^{k}\mathbb{C}^{2}$ then, as above, $f_{z}(u)=f(z\otimes u)$ is a frame
function on $H$. Thus there exists a trace class self-adjoint operator on $H$,
$\alpha(z)$, such that $f_{z}(u)=\left\langle u|\alpha(z)|u\right\rangle $ for
$u$ a state in $H$. Similarly we see that for fixed $u$ a state in $H$,
\[
z\longmapsto\left\langle u|\alpha(z)|u\right\rangle
\]
defines an unentangled frame function on $\mathcal{\otimes}^{k}\mathbb{C}^{2}%
$. Thus there exist for each $u$ a state in $H$ and for $J\subset\Omega_{k}$
we have a function $\xi_{u,J}$ defined by%
\[
\xi_{u,J}(\tau_{J}z)=\sum_{L\subset J}\left\langle u|\alpha(\sigma
_{L}z)|u\right\rangle
\]
for $z$ a state in $\mathcal{\otimes}^{k}\mathbb{C}^{2}$. Now for each fixed
$z$ a state in $\mathcal{\otimes}^{k}\mathbb{C}^{2}$ and $J\subset\Omega_{k}$
\[
u\longmapsto\xi_{u,J}(\tau_{J}z)
\]
is a frame function on $H$. Thus
\[
\xi_{u,J}(\tau_{J}z)=\left\langle u|\phi_{J}(\tau_{J}z)|u\right\rangle
\]
and the rest of the argument is clear. We also note that we can run this
argument backwards using the result with $H=\mathbb{C}$ to prove the
converse (second part of the theorem).
\end{proof}

\begin{remark}
In the above, the frame functions can sum up to an arbitrary constant  $c$
on a UOB. To be a generalized version of a mixed state in quantum mechanics,
as required by eq.~\eqref{ffsum}, we need to set $c=1$. Further, the
definition~\ref{ffdef} imposes the obvious inequalities on  the
sums in $\overset{}{(\ast)}$ that they be non-negative.
\end{remark}

\section{Conclusion} \label{section:conc}

We have classified the unentangled frame functions, first for the
multi-qubit system and then generally when tensor factors are of different
dimensions, including separable Hilbert spaces. The proofs use the theory of
M{\"o}bius functions to explicitly show the combinatorial nature of the
multi-qubit UOB encountered in~\cite{lsw:lduobag}, and the multi-qubit
unentangled frame functions quantify the result of measurements via such
UOB. The structure of the frame functions thus revealed is sufficiently
elegant that we surmise it points to interesting physical interpretations
within the fundamentals of quantum mechanics. Indeed, qubit is the most
basic quantum system, and it is common to see it as a subsystem in quantum
algorithms, and it is also not unusual  to use spatio-temporally separated
measurements, which are by very definition, unentangled. Thus the
information gleaned by such measurements falls within the context of
analysis in this paper. On the other hand, if all the systems being measured
have dimensions at least $3$, then the conclusion of the unentangled
Gleason's~\cite{W:Gleason} theorem applies, which agree with the  original
Gleason's theorem. Another area where we often encounter a mix of systems is
the Hydrogen atom. Its phase space is a spin $\frac{1}{2}$ space tensored
with  $\ell^2(\mathbb{R}^3)$, so it is precisely a sub-case of the  general
case we discuss in Section~\ref{section:genframefun}. These are two of the
more obvious examples, but underline a need to understand  the significance
of unentangled measurements.

%%%%%%%%%%%%%%%%%%%%%%%%%%%%%%%%%%%%%%%%%%%%%%%%%%%%%%%%%%%%%%%%%%%%%%%%%%%%

\begin{acknowledgments} 
The authors would like to thank David Meyer for productive discussions.
\end{acknowledgments}

%%%%%%%%%%%%%%%%%%%%%%%%%%%%%%%%%%%%%%%%%%%%%%%%%%%%%%%%%%%%%%%%%%%%%%%%%%%%


\begin{thebibliography}{99}

\bibitem{vn:mgdq}
  J.\ von Neumann,
  {\it Mathematische Grundlagen der Quantenmechanik},
  J. Springer,
  Berlin, 1932.

\bibitem{h:qmt}
  J.\ Hamhalter,
  {\it Quantum Measure Theory},
  Kluwer Academic Publishers,
  Dordrecht; Boston, 2003.
  

\bibitem{g:mcshs}
  A.\ M.\ Gleason,
  {\it Measures on the closed subspaces of a Hilbert space},
  Journal of Mathematics and Mechanics {\bf 6} (1957),
  no.\ 6,
  885--893.

\bibitem{W:Gleason}
  N.\ R.\ Wallach,
  {\it An unentangled Gleason's theorem},
  Contemporary Mathematics {\bf 305} (2002),
  291--298.

\bibitem{lsw:lduobag}
  J.\ Lebl,
  A.\ Shakeel, and
  N.\ Wallach,
  {\it Local distinguishability of generic unentangled orthonormal bases},
  Physical Review A {\bf 93} (January 2016),
  no\. 1,
  012330/1--6.

\bibitem{r:mobius}
  G.\ Rota,
  {\it On the foundations of combinatorial theory. I. Theory of M\"obius functions},
  Zeitschrift f\"ur Wahrscheinlichkeitstheorie und verwandte Gebiete
  {\bf 2} (1963),
  340--368.
 
\end{thebibliography}
\end{document}